\documentclass{statsoc}
\usepackage{amsmath}
\usepackage{amssymb}

\title[Random Latin squares and Sudoku designs generation]{Random Latin squares and Sudoku designs generation}
\author[R.Fontana]{Roberto Fontana}
\address{Politecnico di Torino,
Dept.of Mathematical Sciences,
Torino,
Italy.}
\email{roberto.fontana@polito.it}

\newtheorem{theorem}{Theorem}

\begin{document}
\begin{abstract}
Uniform random generation of Latin squares is a classical problem. In this paper we prove that both Latin squares and Sudoku designs are maximum cliques of properly defined graphs. We have developed a simple algorithm for uniform random sampling of Latin squares and Sudoku designs. It makes use of recent tools for graph analysis. It has been implemented using SAS.
\end{abstract}

\section{Introduction}
Generating uniformly distributed random Latin squares is a relevant topic. Already in 1933, F. Yates (\cite{yates1933formation}) wrote \begin{quote} ... it would seem theoretically preferable to choose a square at random from all the possible squares of given size. \end{quote} The widely used algorithm for generating random Latin squares of a given order is (\cite{jacobson1996generating}). It is based on a proper set of \emph{moves} that connect all the squares and make the distribution of visited squares \emph{approximately} uniform.

In this paper we present a new approach that is based on the equivalence between Latin squares and maximum cliques of a graph. This approach is also valid for Sudoku designs.

The paper is organized as follows. In Section \ref{sec:graph} the equivalence between Latin squares (Sudoku designs) and maximum cliques of a suitable graph is demonstrated. Section \ref{sec:algo} describes an algorithm for generating uniformly distributed random Latin squares and Sudoku designs. The corresponding SAS code is available in the supporting material. Concluding remarks are made in Section \ref{sec:conclusion}.

\section{Latin squares and Sudoku designs are maximum cliques} \label{sec:graph}

\subsection{Latin squares}
A Latin square of order $n$ is an $n \times n$ matrix $L_n$ in which each of $n$ distinct symbols appear $n$ times, once in each row and one in each column. For the sake of simplicity we consider the integers $1,2,\ldots,n$ as symbols. We denote by $L_n[.,c], c=1,\ldots,n$ ($L_n[r,.],r=1,\ldots,n$) the columns (resp. the rows) of $L_n$ and by $\mathcal{L}_n$ the set of all the Latin squares of order $n$. 

For example, a Latin square of order $4$, $L_4 \in \mathcal{L}_4$, is 
\begin{equation} \label{eq:ls4}
L_4=
\left( 
\begin{array}{cccc}
1&2&3&4\\
4&1&2&3\\
3&4&1&2\\
2&3&4&1
\end{array} 
\right)
\end{equation}

We can look at a Latin square $L_n=[\ell_{rc}; r,c=1,\ldots,n]$ as a set of $n$ disjoint permutation matrices, following the approach recently adopted in (\cite{dahl:2009}) and (\cite{fontana2011fractions}). For each symbol $s$, $s=1,\ldots,n$ we consider the $n \times n$ matrix $P^{(s)}=[p_{rc}^{(s)}; r,c=1,\ldots,n]$ where 
\begin{equation} \label{P0}
p_{rc}^{(s)} = 
\begin{cases}
1 & \textrm{if }  \ell_{rc}=s \\
0 & \textrm{otherwise}   
\end{cases} \;.
\end{equation}
Given a permutation matrix $P$ the corresponding permutation $\pi=(\pi_1,\ldots,\pi_n)$ of $(1,\ldots,n)$ is defined as
\[
\pi = P 1_n
\]
where $1_n$ is the $n \times 1$ column vector whose elements are all equal to $1$.
Viceversa, given a permutation $\pi=(\pi_1,\ldots,\pi_n)$ of $(1,\ldots,n)$ the corresponding permutation matrix $P=[p_{rc}; r,c=1,\ldots,n]$  is defined as
\begin{equation} \label{eq:phi}
p_{rc}= 
\begin{cases}
1 & \textrm{if }  c=\pi_r \\
0 & \textrm{otherwise}   
\end{cases}
\end{equation}
We denote by $\phi$ the function that transform a permutation $\pi$ of $(1,\ldots,n)$ into a permutation matrix $P=\phi(\pi)$ according to Equation \ref{eq:phi}.

For the Latin square $L_4 \in \mathcal{L}_4$ of Equation \ref{eq:ls4}, the permutation matrix $P^{(2)}$ is 

\begin{equation} \label{eq:ls4}
P^{(2)}=
\left( 
\begin{array}{cccc}
0&1&0&0\\
0&0&1&0\\
0&0&0&1\\
1&0&0&0
\end{array} 
\right)
\end{equation}
and the corresponding permutation $\pi^{(2)}$ of $(1,2,3,4)$ is 
\[
\pi^{(2)}=(2,3,4,1) \; .
\]

It immediately follows that a Latin square of order $n$ can be written as 
\begin{equation} \label{eq:ls}
L_n=P^{(1)}+2P^{(2)}+\ldots+nP^{(n)}
\end{equation}
where $P^{(s)}, s=1,\ldots,n$ are \emph{mutually disjoint} permutation matrices. Two permutation matrices $P^{(s)}$ and $P^{(t)}$ are disjoint if and only if $p_{rc}^{(s)} p_{rc}^{(t)} =0$ for each $r,c \in \{1,\ldots,n\}$. Equivalently two permutations $\pi^{(s)}$ and $\pi^{(t)}$ are disjoint if and only if $\pi^{(s)}(r) \neq \pi^{(t)}(r)$ for $r=1,\ldots,n$. 

Without loss of generality, as we will explain below, let us suppose that $P^{(1)}=I_n$ where $I_n$ is the $n \times n$ identity matrix. The permutation $\pi^{(1)}$ corresponding to $P^{(1)}$ is the identity permutation $\iota_n$, $\pi^{(1)}\equiv \iota_n=(1,\ldots,n)$.

Let us denote by $\mathcal{P}_{n}$ the set of all the permutations of $\{1,\ldots,n\}$ and, given $\pi \in \mathcal{P}_{n}$,  by $\mathcal{L}_n^{\pi} \subset \mathcal{L}_n$ the set of all the Latin squares of order $n$ for which $P^{(1)}=\phi(\pi)$ and $P^{(2)} < \ldots < P^{(n)}$ where $P^{(s)} < P^{(t)}$ or, equivalently, $\pi^{(s)} < \pi^{(t)}$ means that $(\pi^{(s)}_1,\ldots, \pi^{(s)}_n) <_{lex} (\pi^{(t)}_1,\ldots, \pi^{(t)}_n)$. The symbol "$<_{lex}$" denotes the standard lexicographic order, $(a_1,\ldots,a_n) <_{lex} (b_1,\ldots,b_n) \Leftrightarrow  \exists m>0 \; \forall i<m \;\; a_i=b_i$ and $a_m < b_m$. For simplicity we will write "$<$" in place of "$<_{lex}$". As it will become clear later on, any order between permutations can be chosen.

Let us consider $\mathcal{L}_n^{\iota_n}$, the set of all the Latin squares of order $n$ for which $P^{(1)}=I_n$ and $P^{(2)} < \ldots < P^{(n)}$. As $\mathcal{L}_n^{\iota_n}$ is built, we can generate all the Latin squares of order $n$, $L_n \in \mathcal{L}_n$, considering
\begin{enumerate}
\item all the $(n-1)!$ permutations $(s_2,\ldots,s_n)$ of the symbols $2,\ldots,n$ and assigning them to the permutation matrices $P^{(2)},\ldots P^{(n)}$ 
\[
I_n+s_2P^{(2)}+\ldots+s_nP^{(n)}
\]
\item all the $n!$ sets $\mathcal{L}_n^{\pi}$ where $\pi\in \mathcal{P}_n$ is a permutation of $(1,\ldots,n)$. We observe that $\mathcal{L}_n^{\pi}$ contains all the Latin squares that are generated permuting the columns of a Latin square $L_n$ of $\mathcal{L}_n^{\iota_n}$ 
\[
\mathcal{L}_n^{\pi}=\{ [ L_n[., \pi_1]|\ldots| L_n[., \pi_n] ]: L_n \in \mathcal{L}_n^{\iota_n} \}
\]  
\end{enumerate}
It follows that in order to generate a random Latin square $L_n$ it is sufficient:
\begin{enumerate}
\item to generate a random Latin square $L_n^{(1)} \in \mathcal{L}_n^{\iota_n}$;
\item to generate $L_n^{(2)}$ by a random permutation of the symbols $2,\ldots,n$ of $L_n^{(1)}$; 
\item to generate $L_n$  by a random permutation of the columns of $L_n^{(1)}$.
\end{enumerate}

We observe that the number $\#\mathcal{L}_n$ of Latin squares of order $n$ is 
\begin{equation} \label{eq:num_ls}
\#\mathcal{L}_n=n!(n-1)!\#\mathcal{L}_n^{\iota_n}
\end{equation} 

To generate a Latin square $L_n \in \mathcal{L}_n^{\iota_n}$ we have to build $n-1$ permutation matrices $P^{(s)}, s=2,\ldots,n$, $P^{(2)} < \ldots < P^{(n)}$, that are mutually disjoint and that are disjoint with $I_n$. In the language of permutations, we have to build $n-1$ \emph{derangements} $\delta^{(s)}, s=2,\ldots,n$ of $(1,\ldots,n)$, $\delta^{(2)} < \ldots < \delta^{(n)}$ such that $\delta^{(s)}(r) \neq \delta^{(t)}(r), r=1,\ldots,n$ for each $s,t \in \{2,\ldots,n\}, s \neq t$. 

Let $\mathcal{D}_{n} \subset \mathcal{P}_{n}$ be the set of all the derangements of $(1,\ldots,n)$: we denote by $d_{n}$ the number of derangements of $(1,\ldots,n)$, $d_{n}=\#\mathcal{D}_{n}$.

Let $\mathcal{G}_{n}=(\mathcal{V}_{n},\mathcal{E}_{n})$ be the undirected graph whose set of vertices $\mathcal{V}_{n}$ is the set of derangements $\mathcal{D}_{n}$ and whose set of edges $\mathcal{E}_{n}$ contains all the couples of derangements $(\delta^{(i)},\delta^{(j)}), i < j$ such that $\delta^{(i)}(r)\neq\delta^{(j)}(r), r=1,\ldots,n$. The following theorem holds.

\begin{theorem} \label{theo}
The Latin squares $L_n$ of order $n$ of $\mathcal{L}_n^{\iota_n}$ are the ordered cliques $C_{n-1}$ of size $n-1$ of $\mathcal{G}_{n}=(\mathcal{V}_{n},\mathcal{E}_{n})$ 
\[
C_{n-1}=(\delta^{(2)},\ldots,\delta^{(n)}), \;\; \delta^{(2)}<\ldots<\delta^{(n)}
\]
$C_{n-1}$ are the largest cliques of $\mathcal{G}_{n}$. 
\end{theorem}
\begin{proof}
A Latin square $L_n \in \mathcal{L}_n^{\iota_n}$ can be written as
\[
L_n = I_n + 2P^{(2)}+\ldots + nP^{(n)}
\]
where $P^{(s)}$ is the permutation matrix corresponding to the derangement $\delta^{(s)}=P^{(s)}1_n, s=2,\ldots,n$ and $\delta^{(2)}<\ldots<\delta^{(n)}$. The derangements $\delta^{(s)}$ are disjoint. It follows that  $\{\delta^{(2)},\ldots,\delta^{(n)}\}$ is a clique of $\mathcal{G}_{n}$.
Viceversa, given the clique $\{\delta^{(2)},\ldots,\delta^{(n)}\}$ with $\delta^{(2)}<\ldots<\delta^{(n)}$ we build
\[
L_n^{\star}=I_n + 2\phi(\delta^{(2)})+\ldots + n\phi(\delta^{(n)})
\]
where $\phi$ is defined in Equation \ref{eq:phi}. It is immediately evident that $L_n^{\star}=L_n$.

Finally, Latin squares correspond to the largest cliques because it is evident that it is not possible to find a set of $m>n$ derangements of $\{1,\ldots,n\}$ that are disjoint.
\end{proof}

\subsection{Sudoku designs}
For the definition of Sudoku designs we refer to \cite{bailey|cameron|connelly:2008}
\begin{quote}
In 1956, W. U. Behrens (\cite{behrens1956}) introduced a specialisation of Latin squares
which he called \emph{gerechte}. The $n \times n$ grid is partitioned into $n$ regions,
each containing $n$ cells of the grid; we are required to place the
symbols $1, . . . , n$ into the cells of the grid in such a way that each symbol
occurs once in each row, once in each column, and once in each region.
The row and column constraints say that the solution is a Latin square,
and the last constraint restricts the possible Latin squares.
By this point, many readers will recognize that solutions to Sudoku puzzles
are examples of \emph{gerechte} designs, where $n = 9$ and the regions are the
$3 \times 3$ subsquares. (The Sudoku puzzle was invented, with the name "number
place", by Harold Garns in 1979.)
\end{quote} 

Analogously to Latin squares, we describe a Sudoku design in terms of Sudoku permutation matrices, \cite{dahl:2009} and \cite{fontana2011fractions}.

Let us define the regions in which the matrix is divided. We will refer to regions as \emph{boxes}. Let us consider a $n \times n$ matrix, where $n=p^2$ and $p$ is a positive integer.  Its row and column positions $(i,j)$ are coded with the integer from $0$ to $p^2-1$. We define boxes $B_{k,m}, \; k,m=0,\ldots,p-1$ as the following sets of positions 
\[
B_{k,m}=\left\{ (i,j):  kp \leq i < (k+1)p , mn \leq j < (m+1)n \right\}
\]
It follows that any $n \times n$ matrix $A$ can be partitioned into submatrices $A_{km}$ corresponding to boxes $B_{k,m}$.

An $n \times n$ matrix $S_n$ is a Sudoku, if in each row, in each column and in each box, each of the integers $1, \ldots, n$ appears exactly once. We denote by $\mathcal{S}_n$ the set of all the $n \times n$ Sudokus. In Sudoku literature, the set of boxes $B_{b,m}, \; m=0,\ldots,p-1$ constitutes the $b^\textrm{th}$ \emph{band}, $b=0,\ldots,p-1$, while the set of boxes $B_{k,s}, \; k=0,\ldots,p-1$ constitutes the $s^\textrm{th}$ \emph{stack}, $s=0,\ldots,p-1$. 

  
Let us define a Sudoku permutation matrix $\tilde{P}$, referred to as an S-matrix $\tilde{P}$, as a permutation matrix of order $n$ which has exactly one $''1''$ in each submatrix $P_{k,m}$ corresponding to boxes $B_{k,m}$, $k,m=0,\ldots,p-1$. Let us denote by $\tilde{\mathcal{P}}_{n} \subset \mathcal{P}_{n}$ the set of all Sudoku permutations. An $n \times n$ Sudoku $S_n$ identifies $n$ matrices $\tilde{P}^{(i)}, \; i=1,\ldots,n$, where $\tilde{P}^{(i)}$ is the S-matrix corresponding to the positions occupied by the integer $i$.
It follows that a Sudoku $S_n \in \mathcal{S}_n$ can be written as
\begin{equation} \label{eq:su}
S_n= \tilde{P}^{(1)} + 2 \tilde{P}^{(2)} + \ldots + n \tilde{P}^{(n)}
\end{equation}
We observe that $\tilde{P}^{(1)},\ldots, \tilde{P}^{(n)}$ are mutually disjoint and that Equation (\ref{eq:su}) is the analogous of Equation (\ref{eq:ls}) for Sudoku designs. 

We observe that the identity permutation $\iota_n$ is not a Sudoku permutation, apart from the trivial $n=2$ case.

We can easily generate S-matrices. Let us define a more compact representation of an S-matrix $S$, by building  the $p \times p$ matrix $S^{\star}$, whose elements are the only possible position, within each box, where $S$ is equal to $1$. 
Among the S-matrices we can define $S_{0,n}^{\star}$ whose elements $(S_{0,n}^{\star})_{km},\; k,m=1,\ldots,p$ are 
\begin{equation*} 
(S_{0,n}^{\star})_{km} = (m,k).
\end{equation*}
For the $n=4$ case we obtain
\begin{equation} \label{P04x4tilde}
S_{0,4}^{\star} =\left[ 
\begin{array}[h]{cc}
(1,1) & (2,1) \\
(1,2) & (2,2)	
\end{array}
\right].
\end{equation}
and the corresponding S-matrix $S_{0,4}$ is 
\begin{equation} \label{P04x4}
S_{0,n} = \left[ 
\begin{array}[h]{cccc}
1 & 0 &  0 &  0 \\
0  & 0  & 1  & 0 \\
0  & 1  & 0 &  0 \\
0  & 0 &  0  & 1	
\end{array}
\right].
\end{equation}
We will denote by $\sigma_{0,n}$ the permutation corresponding to the S-matrix $S_{0,n}$, $\sigma_{0,n}=S_{0,n}1_n$. It will play the role of the identity permutation $\iota_n$ for Sudoku designs.

All the S-matrices can be generated by permuting the rows within each band and the columns within each stack. It follows that the total number of S-matrices is $p!^{2p}$, \cite{dahl:2009}.

Let $\tilde{\mathcal{G}}_{n}=(\tilde{\mathcal{V}}_{n},\tilde{\mathcal{E}}_{n})$ be the undirected graph whose set of vertices $\tilde{\mathcal{V}}_{n}$ is the set of derangements of $\sigma_{0,n}$ \emph{that are also $S$-permutations}, briefly Sudoku-derangements,  and whose set of edges $\tilde{\mathcal{E}_{n}}$ contains all the couple of Sudoku-derangements $(\tilde{\delta}^{(i)},\tilde{\delta}^{(j)}), i < j$ such that $\tilde{\delta}^{(i)}(r)\neq\tilde{\delta}^{(j)}(r), r=1,\ldots,n$.

Theorem \ref{theo} holds if we replace the graph $\mathcal{G}_{n}$ with the graph $\tilde{\mathcal{G}}_{n}$. We observe that $\tilde{\mathcal{G}}_{n}$ is a subgraph of $\mathcal{G}_{n}$.

Let us denote by $\mathcal{S}_n^{\sigma_{0,n}}$ the set of all the Sudokus of order $n$ for which $\tilde{P}^{(1)}=S_{0,n}$ and $\tilde{P}^{(2)}< \ldots < \tilde{P}^{(n)}$. 

As $\mathcal{S}_n^{\sigma_{0,n}}$ is built, we can generate all the Sudokus of order $n=p^2$ considering
\begin{enumerate}
\item all the $(n-1)!$ permutations $(s_2,\ldots,s_n)$ of the symbols $2,\ldots,n$ and assigning them to the permutation matrices $\tilde{P}^{(2)},\ldots \tilde{P}^{(n)}$ 
\[
I_n+s_2\tilde{P}^{(2)}+\ldots+s_n\tilde{P}^{(n)}
\]
\item all the $p!^{2p}$ sets $\mathcal{S}_n^{\sigma}$ where $\sigma$ is a Sudoku permutation of $(1,\ldots,n)$. 
\end{enumerate}
For the total number of Sudokus of order $n$ we get Equation (\ref{eq:num_su}) which is the equivalent of Equation (\ref{eq:num_ls}): 
\begin{equation} \label{eq:num_su}
\# \mathcal{S}_n=(n-1)! p!^{2p} \# \mathcal{S}_n^{\sigma_{0,n}}
\end{equation}

\section{An algorithm for random sampling} \label{sec:algo}
\subsection{Latin squares}


The algorithm takes $n$ as input and gives $L_n$, a random Latin square of order $n$, as output.

The main steps of the algorithm are as follows. 
\begin{enumerate}
\item Build the undirected graph $\mathcal{G}_{n}=(\mathcal{V}_{n},\mathcal{E}_{n})$;
\begin{enumerate}
\item generate $\mathcal{V}_{n} \equiv \mathcal{D}_{n}$, the set of all the derangements $\delta^{(i)}, i=1,\ldots, d_{n}$ of $\{1,\ldots,n\}$;
\item generate $\mathcal{E}_{n}$, the set of all the edges corresponding to all the couples of derangements $(\delta^{(i)},\delta^{(j)}), i < j$ such that $\delta^{(i)}(r)\neq\delta^{(j)}(r), r=1,\ldots,n$.
\end{enumerate}
\item generate all the largest cliques of $\mathcal{G}_{n}$
\item randomly extract one of the largest clique and order its vertices lexicographically. Let use denote this ordered clique by $C_{n-1}=(\delta^{(2)},\ldots,\delta^{(n)})$. The corresponding Latin square is
\[
L_n^{(1)}=I_n + 2\phi(\delta^{(2)})+\ldots + n\phi(\delta^{(n)})
\]
\item randomly choose one permutation $\sigma=(s_2,\ldots,s_n)$ of $(2,\ldots,n)$ and generate
\[
L_n^{(2)}=I_n + s_2\phi(\delta^{(2)})+\ldots + s_n\phi(\delta^{(n)})
\]
\item \label{alg:6} randomly choose one permutation $\gamma$ of $(1,\ldots,n)$ and generate $L_n$ permuting the columns of $L_n^{(2)}$ according to $\gamma$.  
\end{enumerate}

We describe the algorithm for $n=5$.

\begin{enumerate}
\item We generate $\mathcal{D}_{5}$ taking all the permutations $\delta$ of $(1,\ldots,5)$ such that $\delta(r)\neq r$, $r=1,\ldots,5$. $\mathcal{D}_{5}$ contains $44$ derangements. We denote by $\delta^{(i)}, i=1,\ldots,44$ the elements of $\mathcal{D}_{5}$.
\item We generate $\mathcal{E}_{5}$ considering all the $\binom {44}{2}=946$ couples of  derangements $(\delta^{(s)},\delta^{(t)})$, $\delta^{(s)}<\delta^{(t)}$ such that $\delta^{(s)}(r) \neq \delta^{(t)}(r), r=1,\ldots,5$. We find $276$ edges. The graph $\mathcal{G}_{5}$, generated using the function \verb|tkplot| of the R package igraph, (\cite{igraph}), is shown in Figure \ref{fig:g5}.


\begin{figure}
\centering
\makebox{\includegraphics{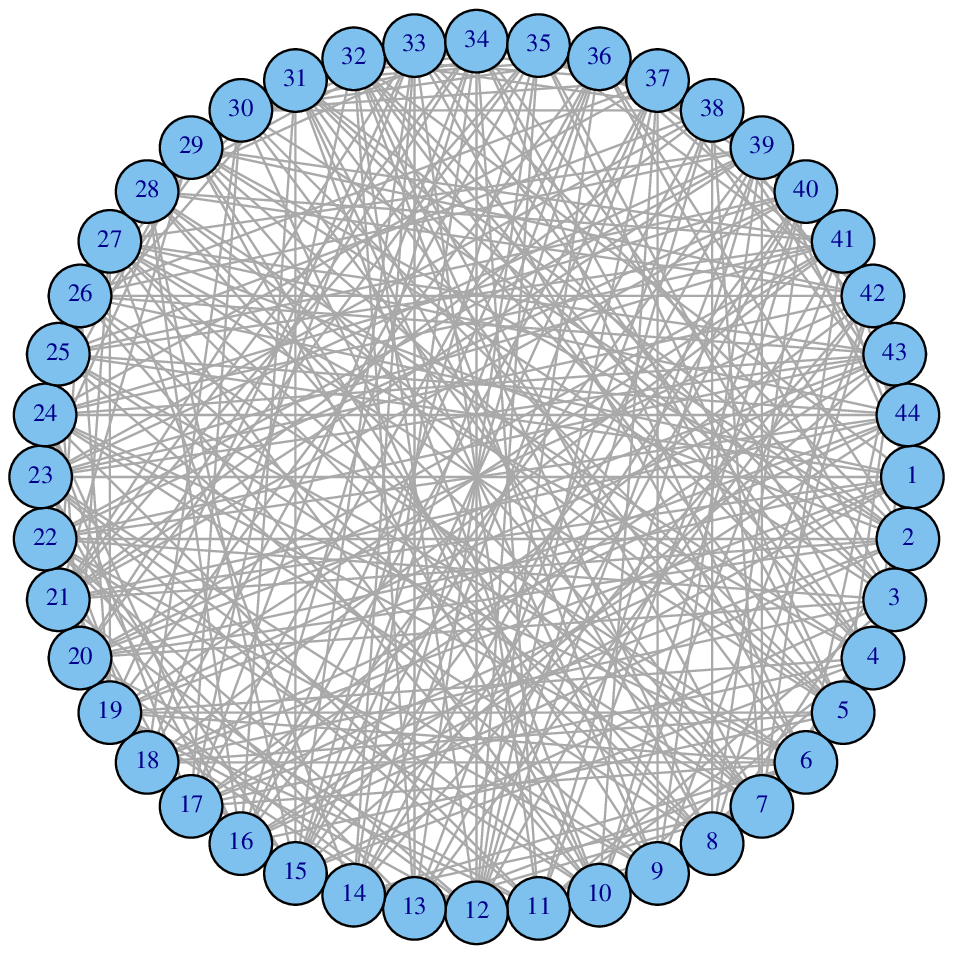}}
\caption{\label{fig:g5} The graph $\mathcal{G}_{5}$}
\end{figure}

\item We use the function \verb|largest.cliques| of the R package igraph, (\cite{igraph}), to get all the largest cliques of $\mathcal{D}_{0,5}$. Equivalently we can use the Optnet procedure of SAS/OR, (\cite{sasor:12}) or Cliquer, \cite{niskanen2003cliquer}. We find $56$ cliques of size $4$. 
\item We randomly choose one clique $C_4$ and we order it lexicographically
 \[
C_4=(\delta^{(11)},\delta^{(17)},\delta^{(23)},\delta^{(37)})
\]
where $\delta^{(11)}=(2, 5, 4, 3, 1)$, $\delta^{(17)}=(3, 4, 5, 1, 2)$, $\delta^{(23)}=(4, 1, 2, 5, 3)$ and $\delta^{(37)}=(5, 3, 1, 2, 4)$. The corresponding Latin square is $L_5^{(1)}=I_n + 2\phi(\delta^{(2)})+3\phi(\delta^{(17)})+4\phi(\delta^{(30)})+5\phi(\delta^{(36)})$, that is
\begin{equation} \label{eq:l5_1}
L_5^{(1)}=
\left( 
\begin{array}{ccccc}
1&2&3&4&5\\
4&1&5&3&2\\
5&4&1&2&3\\
3&5&2&1&4\\
2&3&4&5&1
\end{array} 
\right)
\end{equation}
\item we finally get $L_5$ by randomly choosing one permutation $\sigma$ for the symbols $2,\ldots,5$, $\sigma=(4, 3, 2, 5)$, and one permutation $\gamma$ for the columns $1,\ldots,5$, $\gamma=(3, 1, 2, 4, 5)$  
\begin{equation} \label{eq:l5}
L_5=
\left( 
\begin{array}{ccccc}
3&1&4&2&5\\
5&2&1&3&4\\
1&5&2&4&3\\
4&3&5&1&2\\
2&4&3&5&1
\end{array} 
\right)
\end{equation}

\end{enumerate}

\subsection{Sudoku designs}
The algorithm remains the same apart from the substitution of the graph $\mathcal{G}_{n}$ with $\tilde{\mathcal{G}}_{n}$ and by limiting the permutations of step (\ref{alg:6}) to the permutations of the rows within band and of the columns within stacks. For example, for the case $n=2^2$ we find three maximum cliques of $\tilde{\mathcal{G}}_{4}$. By randomly choosing one permutation of the symbols $2,3,4$ among the six available, and one S-matrix to be used as $\tilde{P}^{(1)}$, among the sixteen available, we can randomly generate one Sudoku from the total of $288$, ($\#\mathcal{S}_4=288$).

\subsection{Computational aspects}
We ran the algorithm using a standard laptop (CPU Intel Core i7-2620M CPU 2.70 GHz 2.70 GHz, RAM 8 Gb). We were able to solve the problems corresponding to the orders up to $n=7$ for which we found $16,942,080$ cliques. For $n=7$ we used Cliquer \cite{niskanen2003cliquer} to find all the cliques. Taking into account symbol and column permutations our algorithm was able to extract uniformly at random a Latin square of order $7$ among all the order $7$ Latin squares that are $7!6!16,942,080=61,479,419,904,000$. 

For Latin squares the number of nodes of $\mathcal{G}_{n}$ coincides with the number of derangements (see Table \ref{tab:der}).
For Sudoku designs the numbers of Sudoku-derangements are $7$ for $n=4$ and $17,972$ for $n=9$.
\begin{table}
\caption{\label{tab:der} Number of derangements of $(1,\ldots,n)$, $n\leq9$}
\centering
\fbox{%
\begin{tabular}{*{10}{c}}
\em 0&\em 1&\em 2&\em 3&\em 4&\em 5&\em 6&\em 7&\em 8&\em 9\\
\hline
1&0&1&2&9&44&265&1854&14833&133496\\
\end{tabular}}
\end{table}

If $n$ becomes large with respect to the available computational resources it is possible to replace the graph $\mathcal{G}_n$ ($\tilde{\mathcal{G}}_n$) with a random subgraph $\mathcal{A}_n^k$ ($\tilde{\mathcal{A}}_n^k$) of it, where $k$ denotes the number of the selected nodes. We point out that, if we take one clique at random from those of the subgraph $\mathcal{A}_n^k$ ($\tilde{\mathcal{A}}_n^k$), the distribution from which we are sampling is not uniform. Anyhow this approach can be useful to select the starting point of the algorithm described in \cite{jacobson1996generating} that is based on moves between different designs. We experimented this approach for the $9 \times 9$ Sudoku, which is the most common structure for the popular Sudoku puzzle.
We randomly chose $809$ Sudoku derangements among the $17,972$ available. The subgraph has $112,579$ edges. Its largest cliques have dimensions equal to $8$ and are $73$. By randomly choosing one clique, one permutation of the symbols $2,\ldots,9$ and one Sudoku matrix we can generate the Sudoku $S_9 \in \mathcal{S}_9$
\begin{equation} 
S_9=
\left( 
\begin{array}{ccccccccc}
1&3&4&5&7&6&2&9&8\\
8&7&2&1&4&9&6&3&5\\
6&9&5&3&2&8&1&7&4\\
7&1&9&8&5&3&4&2&6\\
2&8&6&7&1&4&9&5&3\\
4&5&3&6&9&2&8&1&7\\
3&4&1&9&6&7&5&8&2\\
5&2&8&4&3&1&7&6&9\\
9&6&7&2&8&5&3&4&1 
\end{array} 
\right)
\end{equation}

It is worth noting that recent advances in software for huge graph analysis (millions of nodes) make it possible to manage problems that are extremely interesting from a practical point of view.



\section{Conclusion} \label{sec:conclusion}
This paper presented a simple algorithm for uniform random sampling from the population of Latin squares and Sudoku designs. The algorithm is based on the largest cliques of proper graphs and has been implemented in SAS. The code exports the graph in a format that can be used by other software, like Cliquer, \cite{niskanen2003cliquer}. The algorithm could be run using the entire graph $\mathcal{G}_n$ up to an order $n$ equal to $7$ on a standard pc.

Future research will aim at testing the algorithm for higher orders. Recent advances in graph analytics on huge graph such those arising in social sciences (see e.g. \cite{shao2012managing} for an overview on the subject), make this objective feasible and challenging at the same time. 

\bibliographystyle{chicago}      
\bibliography{fb}
  
\end{document}